\documentclass[review]{elsarticle}

\usepackage{lineno}
\modulolinenumbers[5]

\journal{Journal of Nonlinear Analysis: Hybrid Systems Templates}

\usepackage{moreverb,url}
\usepackage{graphicx}          
\usepackage{natbib}                               
\usepackage{subfigure}
\usepackage{psfrag}
\usepackage{amsmath} 
\usepackage{amssymb}  
\usepackage{bm}
\usepackage{setspace}
\newcommand{\rtn}{\mathbb{R}}
\newcommand{\rtl}{\mathbb{L}}
\newcommand{\ntn}{\mathcal{N}}

\newtheorem{theorem}{Theorem}

\newtheorem{remark}{Remark}

\newtheorem{lemma}{Lemma}

\newtheorem{proof}{Proof}

\begin{document}

\begin{frontmatter}

\title{Edge Agreement of Second-order Multi-agent System with Dynamic Quantization via Directed Edge Laplacian}
\author[mymainaddress]{Zhiwen Zeng}

\author[mymainaddress]{Xiangke Wang \corref{mycorrespondingauthor}}
\cortext[mycorrespondingauthor]{Corresponding author}
\ead{xkwang@nudt.edu.cn}

\author[mymainaddress]{Zhiqiang Zheng}
\author[mysecondaryaddress]{Lina Zhao}

\address[mymainaddress]{College of Mechanics and Automation, National University of Defense Technology, Changsha, Hunan, China}
\address[mysecondaryaddress]{Inner Mongolia Agricultural University, Hohhot, Inner Mongolia, China}

\begin{abstract}
This work explores the edge agreement problem of second-order multi-agent system with dynamic quantization under directed communication. To begin with, by virtue of the directed edge laplacian, we derive a model reduction representation of the closed-loop multi-agent system depended on the spanning tree subgraph. Considering the limitations of the finite bandwidth channels, the quantization effects of second-order multi-agent system under directed graph are considered. Motivated by the observation that the static quantizer always lead to the practical stability rather than the asymptotic stability, the dynamic quantized communication strategy referring to the rooming in-rooming out scheme is employed. Based on the reduced model associated with the \emph{essential edge Laplacian}, the asymptotic stability of second-order multi-agent system under dynamic quantized effects with only finite quantization level can be guaranteed. Finally, simulation results are provided to verify the theoretical analysis.
\end{abstract}


\end{frontmatter}

\linenumbers

\begin{spacing}{1.5}

\section{Introduction}
The coordination control problem of multi-agent system has received increasing amounts of attention recently. Network topology and the information flow have turned out to be an important concern of such issue, as the constraints on communication have a considerable impact on the performance of multi-agent system \cite{xie2013consensus}. Early efforts on such problem focused primarily on the assumption that agents can obtain precise information through local communications as \cite{olfati2004consensus,ren2005consensus}. However, only a finite amount of information data can be transmitted among neighbors at each time instant, since the digital channels are always subject to a limited channel capacity.

To cope with the limitations of the finite bandwidth channels, information data are generally processed by quantizers. Under constrained communication, the multi-agent saddle-point problems are solved by using dual averaging method with quantized information \cite{yuan2014dual}. The spectral properties of the incidence matrix is employed to carry out the convergence analysis of multi-agent system for both the uniform quantizer and the logarithmic quantizer in \cite{dimarogonas2010stability}. Further, in \cite{lavaei2012quantized}, by using the stochastic gossiping algorithm, the explicit relationship between the convergence rate and the communication topology is revealed for the uniform quantizers. Since the above-mentioned static quantizer requires infinite quantization levels which can not be achieved by the realistic digital channels, they always lead to practical stability rather than asymptotic stability, and therefore, the dynamic quantizer with finite quantization level is more of practical significance. In \cite{carli2010quantized}, the coding/decoding strategies based on rooming in-rooming out scheme for the dynamic uniform quantizer is proposed to maintain average consensus and to reach it asymptotically. Based on dynamic encoding and decoding scheme embedded with a scaling function, \cite{li2011distributed} provided an explicit relationship of the asymptotic convergence and the network parameters, especially, the quantization level. In addition, the authors also provide a way to reduce the number of transmitting bits along each digital channel down to merely one bit by designing the control parameters. Most recently, extensions of \cite{li2011distributed} are further discussed in the view of the quantized consensus over directed networks \cite{li2013consensus,li2014quantized,zhang2014distributed}. While these methods are mainly devised for multi-agent system with first-order dynamics, it should be mentioned that second-order multi-agent system may lead to a dramatically different coordination behaviour, even when agents are coupled through similar network topology \cite{yu2010second}. To the best of authors' knowledge, there are still little works to explore the quantization effects on second-order dynamics, especially the dynamic quantization. \cite{li2012distributed} proposes a quantized-observer based encoding-decoding scheme for second-order multi-agent systems with limited information, which shows that expontional asymptotic synchronization can be achieved with 2-bit quantizer for connected graph. The rooming in-rooming out strategy is proposed to achieve asymptotic average consensus for double-integrator multi-agent system with dynamically quantized information transimission in \cite{yu2014asymptotic}. However, the above-mentioned literatures only consider the quantization effects associated with undirected graph, the scenario considering the directed graph is still very challenging, since the quantization may cause undesirable oscillating behavior under directed topology \cite{liu2012quantization}.

In this paper, we are going to deal with the challenging scenario that second-order multi-agent system with dynamic quantization under directed graph. Note that the analysis of the node agreement (consensus problem) has matured, but the work related to the edge agreement \cite{zelazo2011edge,zeng2015convergence} has not been deeply studied yet. Since the quantized measurements bring enormous challenges to the analysis of the synchronization behaviour of the second-order multi-agent system, we are going to explore more details about this term by virtue of the reduced edge agreement model. The main contributions of this paper contain third folders. Firstly, a model reduction representation of the closed-loop multi-agent system is derived based on the observation that the co-spanning tree subsystem can be served as an internal feedback. By utilizing the reduced edge agreement model, the analysis of the whole system can be extremely simplified. In addition, contrary to \cite{li2012distributed} and \cite{yu2014asymptotic}, the quantization effects of second-order multi-agent system under directed communication, rather than undirected topology, is considered. Moreover, by using the the rooming in-rooming out scheme, the asymptotic stability of second-order multi-agent system under dynamic quantized effects can be guaranteed with only finite quantization level.

The rest of the paper is organized as follows: preliminaries and some related notions are proposed in Section 2. The dynamic quantized edge agreement with second-order multi-agent system under directed graph is studied in Section 3. The simulation results are provided in Section 4, while the last section draws the conclusion.





\section{Basic Notions and Preliminary Results}\label{sec:basis}

In this section, some basic notions in graph theory and preliminary results about the synchronization of multi-agent system under quantized information are briefly introduced.

\subsection{Graph and Matrix}

In this paper, we use $\left|  \cdot  \right|$ and $\left\|  \cdot  \right\|$ to denote the Euclidean norm and 2-norm for vectors and matrices respectively. Denote by $I_n$ the identity matrix and by $\bm{0}_{n}$ the zero matrix in $\rtn ^{n\times n}$. Let $\bm{0}$ be the column vector with all zero entries. The null space of matrix $A$ is denoted by $\mathcal{N}(A)$. Let $\mathcal G = \left( {\mathcal V ,\mathcal E} \right)$ be a digraph of order $N$ specified by a node set $\mathcal V$ and an edge set $\mathcal E \subseteq  \mathcal V   \times \mathcal V$ with size $L$. The set of neighbors of node $i$ is denoted by $\ntn_i  = \left\{ {j: e_k = (j,i)  \in \mathcal E } \right\}$. The adjacency matrix of $\mathcal G$ is defined as ${A}_\mathcal{G} = \left[ {a_{ij} } \right] \in \mathbb{R}^{N \times N}$ with nonnegative adjacency elements $a_{ij} > 0 \Leftrightarrow \left( {j, i} \right) \in \varepsilon$. The degree matrix $\Delta_\mathcal{G} = \left[\Delta_{ij}\right]$ is a diagonal matrix with $\left[\Delta_{ii}\right] = \sum\nolimits_{j = 1}^N {a _{ij},i = 1,2, \cdots,N}$, and the graph Laplacian of the weighted digraph $\mathcal{G}$ is defined by $L_\mathcal{G}\left( \mathcal G \right)=\Delta_\mathcal{G}-{A}_\mathcal{G}$， whose eigenvalues will be ordered and denoted as $0 = \lambda _1 \le \lambda _2  \le  \cdots  \le \lambda _{N} $. Denote by $\mathcal{W}(\mathcal{G})$ the $L \times L$ diagonal matrix of $w_k$, for $k=1,2\cdots,L$, where $w_k$ represents the weight of $e_k = (j,i) \in \mathcal{E}$. The incidence matrix $E\left( \mathcal G \right) $ for a digraph is a $\left\{ {0, \pm 1} \right\}$-matrix with rows and columns indexed by nodes and edges of $\mathcal G$ respectively, such that for edge ${e_k}=(j,i) \in \mathcal{E}$, $\left[{E\left( \mathcal G \right)} \right]_{jk} = +1$, $\left[{E\left( \mathcal G \right)} \right]_{ik} = -1$ and $\left[{E\left(\mathcal G \right)} \right]_{lk} = 0$ for $l \ne i,j$.
The in-incidence matrix ${E_{\odot} \left( \mathcal{G} \right)} \in \rtn^{N \times L}$ is a $\{ 0, - 1\}$ matrix with rows and columns indexed by nodes and edges of $\mathcal{G}$, respectively, such that for an edge ${e_k}=(j,i) \in \mathcal{E}$,  $\left[{E_\odot \left( \mathcal G \right)} \right]_{lk} = -1$ for $l=i$, $\left[{E_\odot \left(\mathcal G \right)} \right]_{lk} = 0$ otherwise. The weighted in-incidence matrix $E_ \odot^w(\mathcal G)$ can be defined as $E_ \odot^w(\mathcal G) = {E_{\odot} \left( \mathcal{G} \right)}\mathcal{W}(\mathcal{G})$.
As thus, the graph Laplacian of $\mathcal{G}$ has the following expression \cite{zeng2015edge}:
$L_{\mathcal{G}}({\mathcal G}) = E_\odot^w({\mathcal G}) E({\mathcal G})^T$.
The weighted edge Laplacian of a directed graph $\mathcal{G}$ can be defined as \cite{zeng2015edge}
\begin{align}\label{align:edgelap}
L_e({\mathcal G})  := E({\mathcal G})^T E_\odot^w({\mathcal G}).
\end{align}
A directed path in digraph $\mathcal{G}$ is a sequence of directed edges and a directed tree is a digraph in which, for the root $i$ and any other node $j$, there is exactly one directed path from $i$ to $j$. A spanning tree $\mathcal{G}_{_\mathcal{T}} = \left( {\mathcal V ,\mathcal E_1} \right)$ of a directed graph $\mathcal G = \left( {\mathcal V ,\mathcal E} \right)$ is a directed tree formed by graph edges that connect all the nodes of the graph; a cospanning tree $\mathcal{G}_{_\mathcal{C}}= \left( {\mathcal V , \mathcal E - \mathcal E_1} \right) $ of $\mathcal{G}_{_\mathcal{T}}$ is the subgraph having all the vertices of $\mathcal G$ and exactly those edges that are not in $\mathcal{G}_{_\mathcal{T}}$. Graph $\mathcal{G}$ is called \emph{quasi-strongly connected} if and only if it has a directed spanning tree \cite{Thulasiraman:11b}.
\begin{lemma}{\cite{zeng2015edge}}\label{theorem:Laplacianeigequal}
For any directed graph $\mathcal{G}$, the graph Laplacian $L_{\mathcal{G}}({\mathcal G})$ and the edge Laplacian $L_e({\mathcal G})$ have the same nonzero eigenvalues. If $\mathcal{G}$ is quasi-strongly connected, then the edge Laplacian $L_e({\mathcal G})$ contains exactly $N-1$ nonzero eigenvalues which are all in the open right-half plane.
\end{lemma}

\begin{lemma}{\cite{zeng2015edge}}\label{thm:zeroeigen}
Considering a quasi-strongly connected graph $\mathcal{G}$ of order $N$, the edge Laplacian $L_e({\mathcal G})$ has $L-N+1$ zero eigenvalues and zero is a simple root of the minimal polynomial of $L_e({\mathcal G})$.
\end{lemma}


\subsection{Multi-agent System with Dynamic Uniform Quantization}

The general networked multi-agent system is built upon a group of diffusively coupled linear systems which can be described as follows:
\begin{equation}\label{gensys}
\begin{cases}
\dot x_i(t) = Ax_i(t) + Bu_i(t) \\
y_i(t)  = Cx_i(t) + Dw_i(t)
\end{cases}
\end{equation}
where $x_i(t) \in {\mathbb{R}^n}$ represents the state, $u_i(t) \in {\mathbb{R}^m}$ the controller, $w_i(t)\in {\mathbb{R}^r}$ the exogenous disturbances, $y_i(t)\in {\mathbb{R}^l}$ the locally measured output and $A,B,C$ and $D$ are constant matrices with compatible dimensions. As known, the coupling between each networked agent can be characterised by the communication interconnection topology $\mathcal{G}$. To perform collective behaviours, the networked agents can be naturally modeled by the graph $\mathcal{G}$ with vertices being used to describe agents and the edges being used to represent communication topology. Note that, consider the limited capacity of the practical digital channels, the state values are always quantized, and only finite bits of information can be transmitted via network at each time instant. As thus, the networked multi-agent system with the quantized states can be illustrated as the block diagram in Figure \ref{quan}, in which the connection topology $\mathcal{G}$ is explicitly incorporated into the dynamical system as in \cite{zelazo2013performance}.

\begin{figure}
  \centering
  \includegraphics[width=60mm]{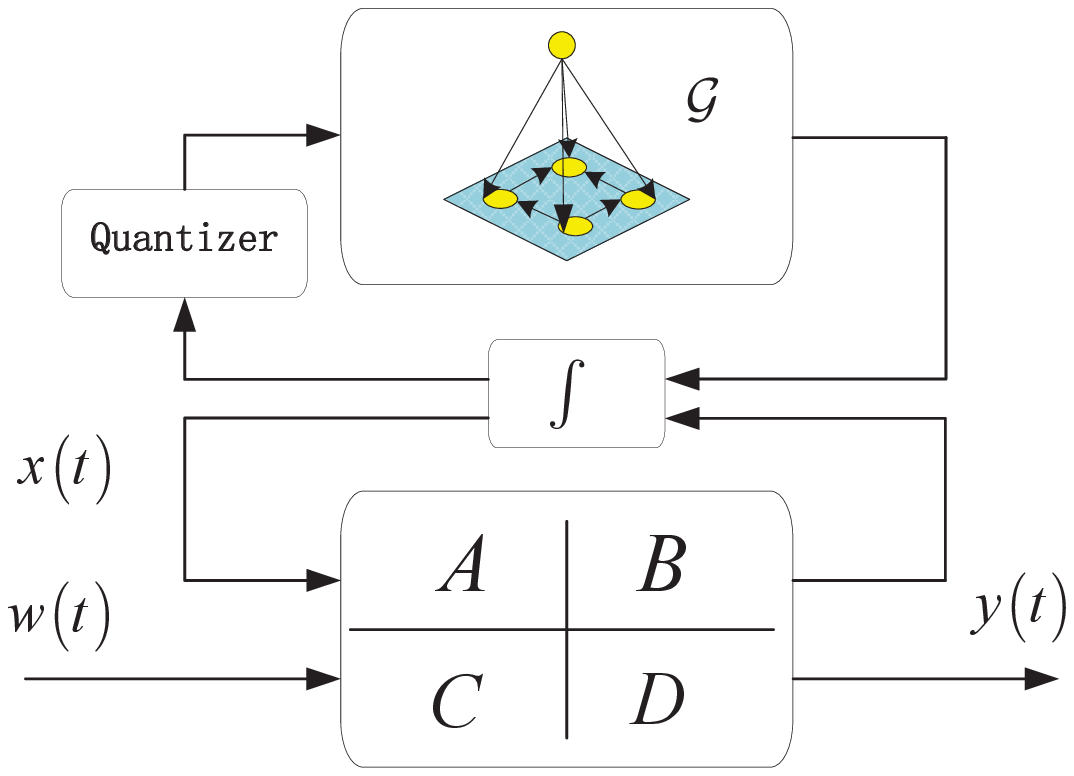}\\
  \caption{The networked multi-agent system with the quantized states under topology $\mathcal{G}$.}\label{quan}
\end{figure}

To realize the quantized communication scheme, a dynamic encoder-decoder pair is employed. Particularly, the quantized information is encoded by the sender side before transmitting and dynamically decoded at the receiver side. Our purpose for this work mainly focuses on the \emph{dynamic uniform quantization}. Suppose $\rtl$ is a finite subset of $\rtn$, then a dynamic uniform quantizer ${q_\mu }:\rtn \to \rtl$ is defined as
\begin{align}\label{dyq}
{q_\mu }\left( {{x_i}} \right)=\mu q_u({{x_i} \over {\mu}})= \left( {\left\lfloor{{{x_i}} \over {\mu \Delta }}\right\rfloor + {1 \over 2}} \right) \mu \Delta
\end{align}
where $q_u$ with quantization error $\Delta$ is a typical uniform quantizer as described in \cite{liu2012quantization}\cite{zeng2015edge}, and $\mu>0$ is adopted as a rooming variable. For a finite-level uniform quantizer with the quantization range ${\mathcal M}$, we have $\left| {{q_\mu }\left( {{x_i}} \right) - {x_i}} \right| \le \mu \Delta$ if $\left| {{x_i}} \right| \le \mu {\mathcal M}$; and $\left| {{q_\mu }\left( {{x_i}} \right)} \right| = \mu \left( {\mathcal{M} - \Delta /2} \right)$, otherwise. For vector $x$, let ${q_\mu }\left( x \right) = {\left[ {{q_\mu }\left( {{x_1}} \right),{q_\mu }\left( {{x_2}} \right), \cdots ,{q_\mu }\left( {{x_N}} \right)} \right]^T}$. Then the following error bounds are hold: $\left| {{q_\mu }\left( x \right) - x} \right| \le \sqrt N \mu \Delta$, if  $\left| {{x_i}} \right| \le \mu \mathcal M$; and  $\left| {{q_\mu }\left( x \right) - x} \right| \le \sqrt N \left( {{\rm{\mathcal M}} - \Delta} \right)$, if  $\left| {{x_i}} \right|  >  \mu \mathcal M$. 
Suppose $\tau$ is a fixed positive number, the rooming variable $\mu$ will be updated at discrete time instants and maintains a constant value on each interval $\left( {k\tau,k\tau + \tau} \right],k \in {\mathbb{Z}_{ \ge 0}}$. As thus, the evolution of the networked multi-agent system with time is discrete rather than continuous in most cases. In fact, by combining the equation \eqref{dyq} with the given system \eqref{gensys} will lead to a \emph{hybrid system}.


\section{Quantized Edge Agreement under Directed Graph}

Considering the quasi-strongly connected graph $\mathcal{G}$ and the most commonly used consensus dynamics \cite{olfati2004consensus}  described as
$\dot x =  -L_{\mathcal{G}} ({\mathcal G})\otimes {I_n} x $,
where $\otimes$ denotes the Kronecker product. Contrary to the most existing works, we study the synchronization problem from the edge perspective by using $L_e$. Following this way, we define the \emph{edge state} vector as
\begin{equation}\label{equation:edgeincidence}
 x_e \left( t \right) = E ({\mathcal G})^T\otimes {I_n} x\left( t \right)
\end{equation}
which represents the difference between the state components of two neighbouring nodes. Taking the derivative of \eqref{equation:edgeincidence} leads to
\begin{align}\label{edge:subsystem}
\dot { x}_e \left( t \right)
 =  - L_e ({\mathcal G})\otimes {I_n} { x_e} \left( t \right)
\end{align}
which is referred as \emph{edge agreement dynamics} in this paper. In comparison to the node agreement (consensus), the edge agreement, rather than requiring the convergence to the agreement subspace, desires the edge dynamics \eqref{edge:subsystem} converge to the origin, i.e., $\mathop {\lim }\nolimits_{t \to \infty } \left| {{x_e}\left( t \right)} \right| = 0$. Essentially, the evolution of an edge state depends on its current state and the states of its adjacent edges. Besides, the edge agreement implies consensus if the directed graph $\mathcal{G}$ has a spanning tree \cite{zelazo2011edge}.

\subsection{Reduced Edge Agreement Model Associated with a Spanning Tree}

 A quasi-strongly connected digraph $\mathcal{G}$ can be rewritten as a union form: $\mathcal{G} = \mathcal{G}_{_\mathcal{T}}  \cup \mathcal{G}_{_\mathcal{C}}$. In addition, according to certain permutations, the incidence matrix $E(\mathcal G)$ can always be rewritten as $ E(\mathcal G) = \left[ {\begin{matrix}{E_{_\mathcal{T} }(\mathcal G)} & {E_{_\mathcal{C}}(\mathcal G)} \end{matrix}} \right]$ as well. Since the cospanning tree edges can be constructed from the spanning tree edges via a linear transformation \cite{zelazo2011edge}, such that $E_{_{\mathcal T}}\left( {{{\mathcal G}}} \right)T({\mathcal G}) = E_{_{\mathcal C}}\left( {\mathcal G} \right)$ with $T({\mathcal G}) = {\left( {E_{_{\mathcal T}}{{\left( {{{\mathcal G}}} \right)}^T}E_{_{\mathcal T}}\left( {{{\mathcal G}}} \right)} \right)^{ - 1}}E_{_{\mathcal T}}{\left( {{{\mathcal G}}} \right)^T}E_{_{\mathcal C}}\left( {{{\mathcal G}}} \right)$ and $rank(E\left(\mathcal{G} \right)) = N-1$ from \cite{Thulasiraman:11b}. We define
\begin{align}\label{R}
R\left( {\mathcal G} \right) = \left[ {\begin{matrix}
   I & {T({\mathcal G})}  \cr
 \end{matrix}}\right]
\end{align}
and then obtain $E\left( {\mathcal G} \right) = E_{_{\mathcal T}}\left( {{{\mathcal G}}} \right)R\left( {\mathcal G} \right)$. The column space of $E(\mathcal{G})^T$ is known as the \emph{cut space} of ${\mathcal G}$ and the null space of $E({\mathcal G})$ is called the \emph{flow space}, which is the orthogonal complement of the cut space.

Before moving on, we introduce the following transformation matrix:
\begin{align*}
{S_e}\left( {\mathcal G} \right) = \left[ \begin{matrix}
   {R{{\left( {\mathcal G} \right)}^T}} & {{\theta_e}\left( {\mathcal G} \right)}  \cr
 \end{matrix}  \right]
\end{align*}
\begin{align*}
{S_e}{\left( {\mathcal G} \right)^{ - 1}} = \left[ {\begin{matrix}
   {{{\left( {R\left( {\mathcal G} \right)R{{\left( {\mathcal G} \right)}^T}} \right)}^{ - 1}}R\left( {\mathcal G} \right)}  \cr
   {{\theta_e}\left( {\mathcal G} \right)^T}  \cr
 \end{matrix}}  \right]
\end{align*}
where ${{\theta_e}\left( {\mathcal G} \right)}$ denote the orthonormal basis of the flow space, i.e., $E\left( {\mathcal G} \right){{\theta_e}\left( {\mathcal G} \right)} = 0$. Since $rank(E\left(\mathcal{G} \right)) = N-1$, one can obtain that $dim({\theta_e}\left( {\mathcal G} \right))= \mathcal{N}(E\left(\mathcal{G} \right))$ and ${\theta_e}\left( {\mathcal G} \right)^T{\theta_e}\left( {\mathcal G} \right)=I_{L-N+1}$. 

Make use of the following transformation for \eqref{edge:subsystem}:
\begin{align*}
S_e^{ - 1}{x_e(t)} =  \left( {\begin{matrix}
    {{x_{_{\mathcal T}}}(t)}  \cr
    \bm{0}  \cr
 \end{matrix} } \right)
\end{align*}
where $x_{_{\mathcal T}} = E_{_{\mathcal T}}({\mathcal G})^T x\left( t \right)$ represents the states across a specific spanning tree of $\mathcal{G}$. Then one can obtain a reduced model representation of \eqref{edge:subsystem} as follows£º
\begin{align}\label{xtsub}
\dot{x}_{_\mathcal{T}} = -E_{_\mathcal{T}} ({\mathcal G})^T E_\odot^w ({\mathcal G}){R{{\left( {\mathcal G} \right)}^T}}\otimes {I_n}{x_{_{\mathcal T}}}\left( t \right)
\end{align}
which captures the dynamical behaviour of the whole system. We refer ${{\hat L}_e(\mathcal G)} = E_{_{\mathcal T}}(\mathcal G)^T{E_ \odot^w (\mathcal G)}{R(\mathcal G)^T}$ as the \emph{essential edge Laplacian} and then we have the following lemma.

\begin{lemma}[\cite{zeng2015edge}]\label{essLe}
The essential edge Laplacian ${{\hat L}_e(\mathcal G)}$ has the same eigenvalues of $L_e(\mathcal G)$ except the zero eigenvalues.
\end{lemma}

It's clear that, by applying the above-mentioned similar transformation will lead to
\begin{align}\label{tran}
{S_e}{\left( {\mathcal G} \right)^{ - 1}}{L_e(\mathcal G)}{S_e}\left( {\mathcal G} \right) = \left[ {\begin{matrix}
   {{\hat L}_e(\mathcal G)} & E_{_{\mathcal T}}^T(\mathcal G){E_ \odot^w (\mathcal G)}{{\theta_e}\left( {\mathcal G} \right)}  \cr
   \bm{0} & \bm{0} \cr
 \end{matrix}}  \right].
 \end{align}
Then the eigenvalues of the block matrix are the solution of
$$\lambda^{(L-N+1)} \det \left( {\lambda I - {{\hat L}_e}(\mathcal G)} \right) = 0$$
which shows that ${{\hat L}_e}(\mathcal G)$ has exactly all the nonzero eigenvalues of $L_e(\mathcal G)$.
Meanwhile, we can construct the following Lyapunov equation as
\begin{align}\label{edgelyap}
H{{\hat L}_e}(\mathcal G) + \hat L_e(\mathcal G)^TH = {I_{N - 1}}
\end{align}
where $H$ is a positive definite matrix.

The weighted in-incidence matrix can be represented as $ E_\odot^w({\mathcal G}) = \left[ {\begin{matrix}{ E_{\odot_\mathcal{T}}^w }({\mathcal G}) & { E_{\odot_\mathcal{C}}^w}({\mathcal G}) \cr \end{matrix}} \right]$ according to $\mathcal{G} = \mathcal{G}_{_\mathcal{T}}  \cup \mathcal{G}_{_\mathcal{C}}$. From \eqref{xtsub}, one can obtain
\begin{align*}
{{\dot x}_{_{\mathcal T}}}\left( t \right) 
= & (- {L_e^{_\mathcal{T}}} ({\mathcal G})- { E_{_\mathcal{T}} ({\mathcal G})^T{ E_{\odot_\mathcal{C}}^w}({\mathcal G})}{T({\mathcal G})^T})\otimes {I_n}{x_{_{\mathcal T}}}\left( t \right)\nonumber\\
\end{align*}
where $L_e^{_\mathcal{T}} ({\mathcal G}) = E_{_\mathcal{T}} ({\mathcal G})^T E_{\odot{_\mathcal{T}}}^w ({\mathcal G})$. Since   $E_{_{\mathcal T}}\left( {{{\mathcal G}}} \right)T({\mathcal G}) = E_{_{\mathcal C}}\left( {\mathcal G} \right)$ as mentioned before, the co-spanning tree states can be reconstructed through the matrix $T({\mathcal G})$ as
\begin{align*}
x_{_{\mathcal C}} = E_{_{\mathcal C}}({\mathcal G})^T\otimes {I_n} x\left( t \right) =  {T({\mathcal G})^T}\otimes {I_n}{x_{_{\mathcal T}}}\left( t \right).
\end{align*}
Therefore, the co-spanning tree states can be viewed as an internal feedback on the edges of the spanning tree subgraph shown in Figure \ref{figure:internalfb}.

\begin{figure}[hbtp]
\centering
{\includegraphics[height=3.0cm]{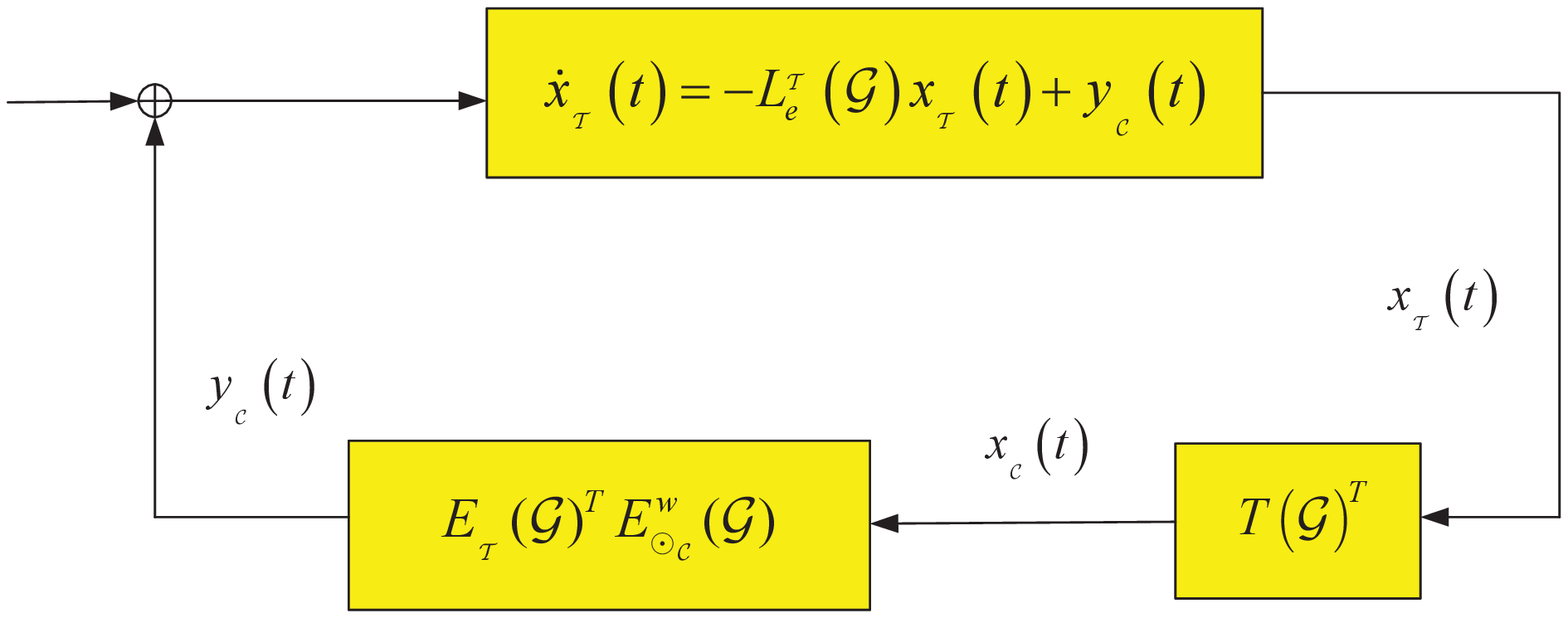}}
\caption{The co-spanning tree states can serve as an internal feedback state.}
\label{figure:internalfb}
\end{figure}

\subsection{Main Result and Stability Analysis}
We consider a group of N networked agents and the dynamics of the $i$-th agent is represented by
\begin{align}\label{dynam}
\begin{cases}
 \dot x_i(t)  = v_i(t)  \\
 \dot v_i(t)  = u_i(t)
\end{cases}
\end{align}
where $x_i(t)  \in \rtn^n$ is the position, $v_i(t)  \in \rtn^n$ is the velocity and $u _i(t)  \in \rtn^n$ is the control input. The goal for designing distributed control law $u_i(t)$ is to synchronize the velocities and positions of the $N$-networked agents.

The generally studied second-order consensus protocol proposed in \cite{yu2010second} is described as follows: $ u_i(t) =  \alpha \sum\limits_{j \in \mathcal{N}_i}^N {a_{ij}}\left( {{x_j}\left( t \right) - {x_i}\left( t \right)} \right) + \beta \sum\limits_{j \in \mathcal{N}_i}^N {{a_{ij}}\left( {{v_j}\left( t \right) - {v_i}\left( t \right)} \right)}$, for $i =1,2\cdots,N$, where $\alpha > 0$ and $\beta > 0$ are the coupling strengths. As in \cite{dimarogonas2010stability}, we assume that each agent $i$ has only quantized measurements of the relative position $ q_\mu \left( {{x_i} - {x_j}} \right)$ and velocity information $q_\mu \left( {{v_i} - {v_j}} \right)$. In that way, the protocol can be modified as
\begin{align}\label{quantizedpro}
u_i(t) = & \alpha \sum\limits_{j \in \mathcal{N}_i}^N {a_{ij}} q_\mu \left( {{x_j}\left( t \right) - {x_i}\left( t \right)} \right)  \nonumber\\
& + \beta \sum\limits_{j \in \mathcal{N}_i}^N {{a_{ij}} q_\mu  \left( {{v_j}\left( t \right) - {v_i}\left( t \right)} \right)}
\end{align}
for $i =1,2\cdots,N$.
To ease the notation, we simply use $E$, $E_ \odot^w$ and $L_e$ instead of  $E(\mathcal{G})$, $E_ \odot^w(\mathcal{G})$ and $L_e(\mathcal{G})$ in the following parts.

Considering the dynamics of the networked agents as describing in \eqref{dynam}, by directly applying the quantized protocol \eqref{quantizedpro}, we obtain
\begin{align}
\begin{cases}
{{\dot x}_i}\left( t \right) = {v_i}\left( t \right) \\
{{\dot v}_i}\left( t \right) =   \alpha \sum\limits_{j \in \mathcal{N}_i}^N {a_{ij}}q_\mu \left( {{x_j}\left( t \right) - {x_i}\left( t \right)} \right) \nonumber \\
~~~~~~~ +  \beta \sum\limits_{j \in \mathcal{N}_i}^N {{a_{ij}}q_\mu \left( {{v_j}\left( t \right) - {v_i}\left( t \right)} \right)}.
\end{cases}
\end{align}
To ease the analysis, we technically choose $\alpha = \sigma^2$ and $\beta =\sigma^3$ ($\sigma > 0$) as in \cite{hu2012second}. Then the system can be collected as
\begin{equation}\label{dyn1}
\begin{cases}
\dot x\left( t \right)  =  v\left( t \right)  \\
\dot v\left( t \right)  =   -  \sigma^2 {E_ \odot^w } \otimes {I_n}q_\mu \left( {{E^T}  \otimes {I_n}x\left( t \right)} \right)  \\
~~~~~~- \sigma^3 {E_ \odot^w } \otimes {I_n}q_\mu \left( {{E^T} \otimes {I_n}v\left( t \right)} \right)
\end{cases}
\end{equation}
with $x(t)$, $v(t)$ denoting the column stack vector of ${x_i(t)}$  and ${v_i(t)}$ respectively.



By left-multiplying $E^T\otimes {I_n}$ of both sides of \eqref{dyn1}, we obtain
\begin{equation}\label{secedgesystem}
\begin{cases}
{{\dot x}_e} = {v_e}  \\
{{\dot v}_e} =  - \sigma^2 {L_e} \otimes {I_n}q_\mu \left( {{x_e}} \right) -\sigma^3 {L_e} \otimes {I_n}q_\mu \left( {{v_e}} \right)
\end{cases}
\end{equation}
with $x_e = E^T \otimes {I_n} x$, $v_e = E^T\otimes {I_n} v$.

The quantization error satisfies $\left| e_{{x_e}}\right|, \left|e_{{v_e}} \right| \le  \sqrt {nL} {\mu}{\Delta}$, where ${e_{{x_e}}} = q_\mu\left( {{x_e}} \right) - {x_e}$ and ${e_{{v_e}}} = q_\mu\left( {{v_e}} \right) - {v_e}$. Then dynamic system \eqref{secedgesystem} can be written as the following form:
\begin{equation}\label{edgesys}
\begin{cases}
{{\dot x}_e}\left( t \right) = {v_e}\left( t \right)  \\
{{\dot v}_e}\left( t \right) =  - \sigma^2 {L_e} \otimes {I_n}{x_e} -\sigma^3 {L_e} \otimes {I_n}{v_e} \\
~~~~~~~~~~  - \sigma^2 {L_e} \otimes {I_n}{e_{{x_e}}} - \sigma^3 {L_e} \otimes {I_n}{e_{{v_e}}}.
\end{cases}
\end{equation}

Let
$z = \left[ {\begin{matrix}
   x_e^T \quad v_e^T
  \end{matrix}} \right]^T$
and
 ${\omega} = \left[ {\begin{matrix}
   {e_{x_e}}^T \quad {e_{v_e}}^T
  \end{matrix}} \right]^T$, then system \eqref{edgesys} can be recasted in a compact matrix form as follows:
\begin{align}\label{iniz}
\dot z =  {\mathcal L} \otimes {I_n}z + {{\mathcal L}_1} \otimes {I_n}{\omega}
\end{align}
with
$
 {\mathcal L} = \left[ {\begin{matrix}
   {{\bm{0}_L}} & {{I_L}}  \cr
   { - \sigma^2 {L_e}} & { - \sigma^3 {L_e}}  \cr
 \end{matrix} } \right]$ and $
{{\mathcal L}_1} = \left[ {\begin{matrix}
   {{\bm{0}_L}} & {{\bm{0}_L}}  \cr
   { - \sigma^2 {L_e}} & { - \sigma^3 {L_e}}  \cr
 \end{matrix} } \right]
$, where $\left|\omega \right|\le \sqrt {2nL} {\mu}{\Delta}$.

Since $L_e$ contains zero eigenvalues, the direct analysis of \eqref{iniz} is difficult. However, the reduced edge agreement model will be of great help in this scene. To begin with, we make use of the following transformation
\begin{align*}
S_e^{ - 1}\otimes {I_n}{x_e} =  \left( {\begin{matrix}
    {{x_{_{\mathcal T}}}}  \cr
    \bm{0}  \cr
 \end{matrix} } \right)~~
 S_e^{ - 1}\otimes {I_n}{v_e} =  \left( {\begin{matrix}
    {{v_{_{\mathcal T}}}}  \cr
    \bm{0}  \cr
 \end{matrix} } \right)
\end{align*}
\begin{align*}
S_e^{ - 1}\otimes {I_n}{e_{x_e}} =
 \left[ {\begin{matrix}
   {{{\left( {RR^T} \right)}^{ - 1}}R\otimes {I_n}{e_{{x_e}}}}  \cr
   {{\theta_e}^T\otimes {I_n}{e_{{x_e}}}}  \cr
 \end{matrix}} \right]
\end{align*}
\begin{align*}
S_e^{ - 1}\otimes {I_n}{e_{v_e}} = \left[ {\begin{matrix}
   {{{\left( {RR^T} \right)}^{ - 1}}R\otimes {I_n}{e_{{v_e}}}}  \cr
   {{\theta_e}^T\otimes {I_n}{e_{{v_e}}}}  \cr
 \end{matrix}} \right].
\end{align*}
Then we define $z_{_{\mathcal T}} = \left[ {\begin{matrix}
   x^T_{_{\mathcal T}} & v^T_{_{\mathcal T}}  \cr
  \end{matrix}} \right]^T$ and ${\hat L}_{_{\mathcal O}} = E_{_\mathcal{T}}^T E_\odot^w$. Finally, system \eqref{edgesys} can be rewritten into
\begin{align}\label{STsubsys}
{{\dot z}_{_{\mathcal T}}} =   {{\mathcal L}_{_{\mathcal T}}} \otimes {I_n}{z_{_{\mathcal T}}} + { {{\mathcal L}_{_{\mathcal T 1}}}} \otimes {I_n}{\omega}
\end{align}
with ${{\mathcal L}_{_{\mathcal T}}} = \left[ {\begin{matrix}
   {{\bm{0}_{N - 1}}} & {{I_{N - 1}}}  \cr
   { - \sigma^2 {{\hat L}_e}} & { - \sigma^3 {{\hat L}_e}}  \cr
  \end{matrix} }  \right]$, ${{\mathcal L}_{_{{\mathcal T}1}}} = \left[ {\begin{matrix}
   {{\bm{0}_{(N-1\times L) }}} & {{\bm{0}_{(N-1 \times L)}}}  \cr
   { - \sigma^2 {{\hat L}_{_{\mathcal O}} }} & { - \sigma^3 {{\hat L}_{_{\mathcal O}} }}  \cr
  \end{matrix}} \right]$.
To further explore the quantization effects on the edge agreement, we propose the following theorem.

\begin{theorem}\label{mainthm}
Considering the quasi-strongly connected digraph $\mathcal G$ associated with the edge Laplacian $L_e$ , suppose $\mathcal{Q} =  - \left( {{\mathcal P}{{\mathcal L}_{_{\mathcal T}}} + {\mathcal L}_{_{\mathcal T}}^T{\mathcal P}} \right)$ with  ${\mathcal P} = \left[ {\begin{matrix}
   {\sigma H} & {{H}}  \cr
   {{H}} & {\sigma H}  \cr
   \end{matrix}} \right]$, where $H$ is obtained by \eqref{edgelyap}. Assume that $\mathcal{M}$ is large enough compared to $\Delta$, so that we have
\begin{align}\label{cond1}
\sqrt {{{{\lambda _{\min }}\left(\mathcal{P} \right)} \over {{\lambda _{\max }}\left( \mathcal{P} \right)}}} \mathcal{M} > 2\Delta \max \left\{ {1,{{\sqrt {2nL} \left\| {{\mathcal P}{{\mathcal L}_{_{{\mathcal T}1}}} } \right\|} \over {{\lambda _{\min }}\left( \mathcal{Q} \right)}}} \right\}.
\end{align}
Then there exists a hybrid quantized feedback control policy \eqref{quantizedpro} that makes the edge agreement of \eqref{STsubsys} asymptotically achieved.
\end{theorem}
\begin{proof}
By selecting
\begin{align*}
\sigma > \sqrt {{{{\lambda _{\max }(H)}} \over 2} + 1}
\end{align*}
then $\mathcal{P}$ and
\begin{align*}
\mathcal{Q} =  - \left( {{\mathcal P}{{\mathcal L}_{_{\mathcal T}}} + {\mathcal L}_{_{\mathcal T}}^T{\mathcal P}} \right) = \left[ \begin{matrix}
   {\sigma^2{I_{N - 1}}} & {{\sigma^3}{I_{N - 1}} - \sigma H}  \cr
   {{\sigma^3}{I_{N - 1}} - \sigma H} & {{\sigma^4}{I_{N - 1}}- 2H}  \cr
  \end{matrix} \right].
\end{align*}
are positive definite according to Schur complements theorem \cite{yu2010second}.

For the edge Laplacian dynamics \eqref{STsubsys}, we can choose the following Lyapunov function candidate:
\begin{align}\label{ly}
V\left( {{z_{_{\mathcal T}}}} \right) = z_{_{\mathcal T}}^T{\mathcal P} \otimes {I_n}{z_{_{\mathcal T}}}.
\end{align}


By taking the derivative of \eqref{ly} along the trajectories of \eqref{STsubsys}, we have
\begin{align}\label{dlyap}
\dot V\left( {{z_{_{\mathcal T}}}} \right) 
=  & -z_{_{\mathcal T}}^T\mathcal{Q}\otimes {I_n}{z_{_{\mathcal T}}} + z_{_{\mathcal T}}^T{\mathcal P}{{\mathcal L}_{_{{\mathcal T}1}}}\otimes {I_n}{\omega} + {\omega}^T{{\mathcal L}_{_{{\mathcal T}1}}^T}{\mathcal P}\otimes {I_n}z_{_{\mathcal T}} \nonumber\\
& \le  - \lambda_{\min } \left( \mathcal Q \right){\left| {{z_{_{_{\mathcal T}}}}} \right|^2} + 2\sqrt {2nL} {\mu \Delta}\left\| {{\mathcal P}{{\mathcal L}_{_{{\mathcal T}1}}} } \right\|\left| {{z_{_{\mathcal T}}}} \right| \nonumber \\
& = -\left| {{z_{_{_{\mathcal T}}}}} \right| \lambda_{\min } \left(\mathcal Q \right)\Bigl( \left| {{z_{_{_{\mathcal T}}}}} \right| - \Theta {\mu \Delta}  \Bigr)
\end{align}
in which $\Theta  =  2\sqrt {2nL}\left\| {{\mathcal P}{{\mathcal L}_{_{{\mathcal T}1}}} } \right\|/\lambda_{\min } \left(\mathcal Q \right) > 0$.

Based on lemma 1 of \cite{liberzon2003hybrid}, for an arbitrary $\varepsilon > 0$, we can define the ellipsoids $${R_1}: = \left\{ {{z_{_{\mathcal T}}}:z_{_{\mathcal T}}^T{\mathcal P} \otimes {I_n}{z_{_{\mathcal T}}} \le {\lambda _{\min }}\left( \mathcal{P} \right){{\mathcal M}^2}{\mu ^2}} \right\}$$ and
$${R_2}: = \left\{ {{z_{_{\mathcal T}}}:z_{_{\mathcal T}}^T{\mathcal P} \otimes {I_n}{z_{_{\mathcal T}}} \le {\lambda _{\max }}\left( \mathcal{P} \right){\Theta ^2}{\Delta ^2}{{\left( {1 + \varepsilon } \right)}^2}{\mu ^2}} \right\}.$$
According to  \eqref{dlyap}, $R_1$ and $R_2$ are invariant regions for multi-agent system \eqref{STsubsys}. With this setting, the trajectories of \eqref{STsubsys} starting in ${R_1}$ will approach $R_2$ in finite time.

Between ellipsoids ${R_1}$ and $R_2$, we have the following formula:
\begin{align*}
\mathcal{M} \mu \ge \left| z_{_{\mathcal T}} \right| \ge \left( {1 + \varepsilon } \right)\Theta \mu \Delta
\end{align*}
which implies
\begin{align}
\dot V \le  - {\lambda _{\min }}\left( \mathcal{Q} \right){\varepsilon  \over {1 + \varepsilon }}{\left| z_{_{\mathcal T}} \right|^2} \le  - {{{\lambda _{\min }}\left( \mathcal{Q} \right)} \over {\lambda_{\max} \left( \mathcal{P} \right)}}{\varepsilon  \over {1 + \varepsilon }}V.
\end{align}
Let $\alpha = { {{\lambda _{\min }}\left( \mathcal{Q} \right)} \varepsilon  / {\lambda_{\max} \left( \mathcal{P} \right)} ({1 + \varepsilon })}$, then by applying the Comparison Lemma \cite{khalil2002nonlinear}, we can provide the following estimates of the convergence rate as in \cite{yu2014asymptotic}: \begin{align*}
V\left( z_{_{\mathcal T}}(t)  \right)
& \le e^{-\alpha t} V\left( {{z_{_{\mathcal T}}}(0)}  \right)
\end{align*}
with the estimation of the upper bound of the convergence
time that starting in $R_1$ enter $R_2$ as
\begin{align}\label{T}
T = {1 \over \alpha }\ln {{{\lambda _{\min }}\left( \mathcal{P} \right){{\mathcal M}^2}} \over {{\lambda _{\max }}\left(\mathcal{P} \right){\theta ^2}{\Delta ^2}{{\left( {1 + \varepsilon } \right)}^2}}}.
\end{align}

To guarantee the asymptotic stability of the whole system, the Liberzon¡¯s design strategy \cite{liberzon2003hybrid} is employed, in which  the control scheme contains two folders: ``rooming out'' to detect the measurement of states by increasing $\mu$; ``rooming in'' to achieve more accurate quantization by decreasing $\mu$. 

\emph{Rooming out}. Firstly, we initialize $u_i = 0$ and let $\mu(0)= 1$. By increasing $\mu$ fast enough to dominate the rate of growth of $\left|e^{At}\right|$, then we can pick a time $t_0$ such that
\begin{align*}
\left| {{q_\mu }\left( {{z_{_{\mathcal T}}}\left( t_0 \right)} \right)} \right| \le \sqrt {{{{\lambda _{\min }}\left( {\mathcal P} \right)} \over {{\lambda _{\max }}\left( {\mathcal P} \right)}}} {\mathcal M}\mu \left( t_0 \right) - \Delta \mu \left( t_0 \right).
\end{align*}
Therefore, we can obtain
\begin{align*}\left| {{{z_{_{\mathcal T}}\left( {{t_0}} \right)} \over {\mu \left( {{t_0}} \right)}}} \right| \le \sqrt {{{{\lambda _{\min }}\left( {\mathcal P} \right)} \over {{\lambda _{\max }}\left( {\mathcal P} \right)}}} \mathcal{M}
\end{align*}
which implies that $z_{_{\mathcal T}}(t_0)$ belongs to the ellipsoid $R_1(\mu(t_0))$, and this event can be detected using only available quantized measurements.

\emph{Rooming in}. When the initial state is in ellipsoid $R_1$ with the initial rooming variable $\mu(t_0)$, the zooming-in phase starts with the update interval $T$. Let $\mu(t) = \mu(t_0)$ for $t \in \left[ {{t_0},{t_0} + T} \right)$, where $T$ is given by \eqref{T}. Then $x(t_0+T)$ belongs to the ellipsoid $R_2$. Let the rooming in rule is as
\begin{align}
\mu  = {\Omega ^k}{\mu _0},~\Omega  = {{\sqrt {{\lambda _{\max }}\left( \mathcal{P} \right)} \Theta \Delta \left( {1 + \varepsilon } \right)} \over {\sqrt {{\lambda _{\min }}\left( \mathcal{P} \right)} {\mathcal M}}}
\end{align}
for $t \in \left[ {kT,\left( {k + 1} \right)T} \right]$ where $T$ is defined as \eqref{T} and $k$ is the number of update times. According to \eqref{cond1}, it's easy to check $\Omega < 1$ and $\mu(t_0+T) < \mu(t_0)$. To decrease $\mu$ by means of multiplying it by the scaling factor $\Omega$, we have $\mu(t) \to 0 $ which also implies ${z_{_{\mathcal T}}}(t) \to 0$.
\end{proof}

\begin{remark}
With the quantization range $\mathcal{M}$, the quantizer obtains $(2\mathcal{M}+1)$ quantization levels. In addition, only $\lceil {{{\log }_2}\left( {2\mathcal{M}} \right)} \rceil $ bits are required while transmitting data at each time interval.
\end{remark}

\begin{remark}
While the initial state is unknown, the open-loop rooming out stage is utilized to guarantee the state of the system can be adequately measured. As the initial states of multi-agent system are generally known for quantizers, we can select a suitable rooming variable $\mu_0$ in advance to keep the system starts in the ellipsoids $R_1$ without the rooming out stage.
\end{remark}

\section{Simulation}
Consider the multi-agent system consisting of a group of $5$ agents associated with a quasi-strongly connected graph as shown in Fig. \ref{span}, where ${e_1},{e_2},{e_3},{e_4} \subset {{\mathcal G}_{_{\mathcal T}}}$ and ${e_5} \subset {{\mathcal G}_{_{\mathcal C}}}$.
\begin{figure}[hbtp]
\centering
{\includegraphics[height=3cm]{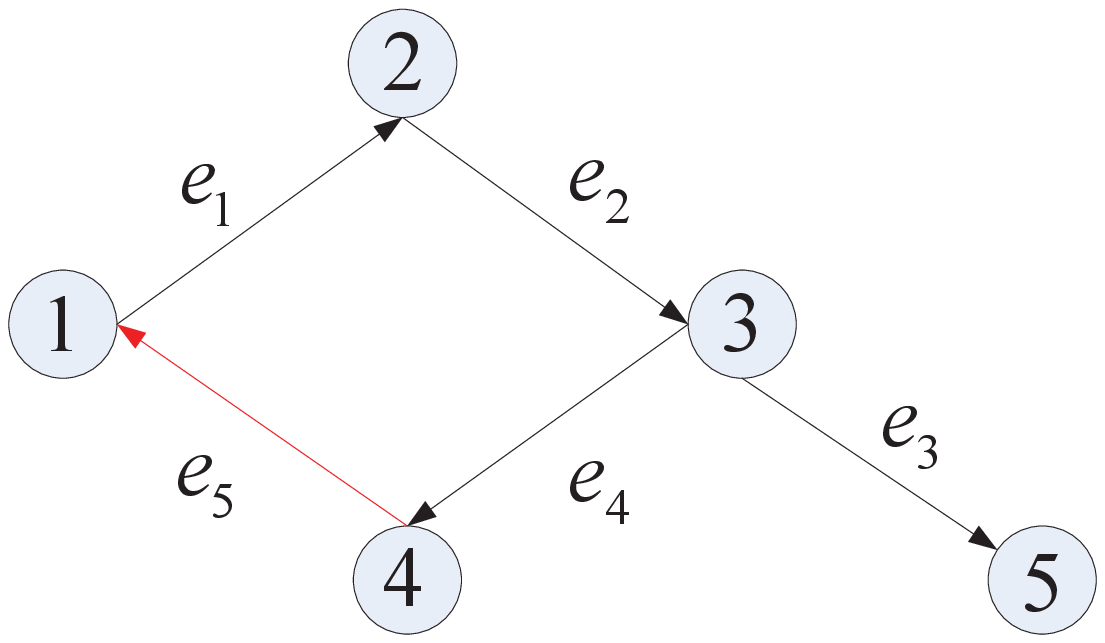}}
\caption{A quasi-strongly connected graph of $5$ agents.}
\label{span}
\end{figure}

\begin{figure*}[hbtp]
\begin{center}
\mbox{
{\includegraphics[width=65mm]{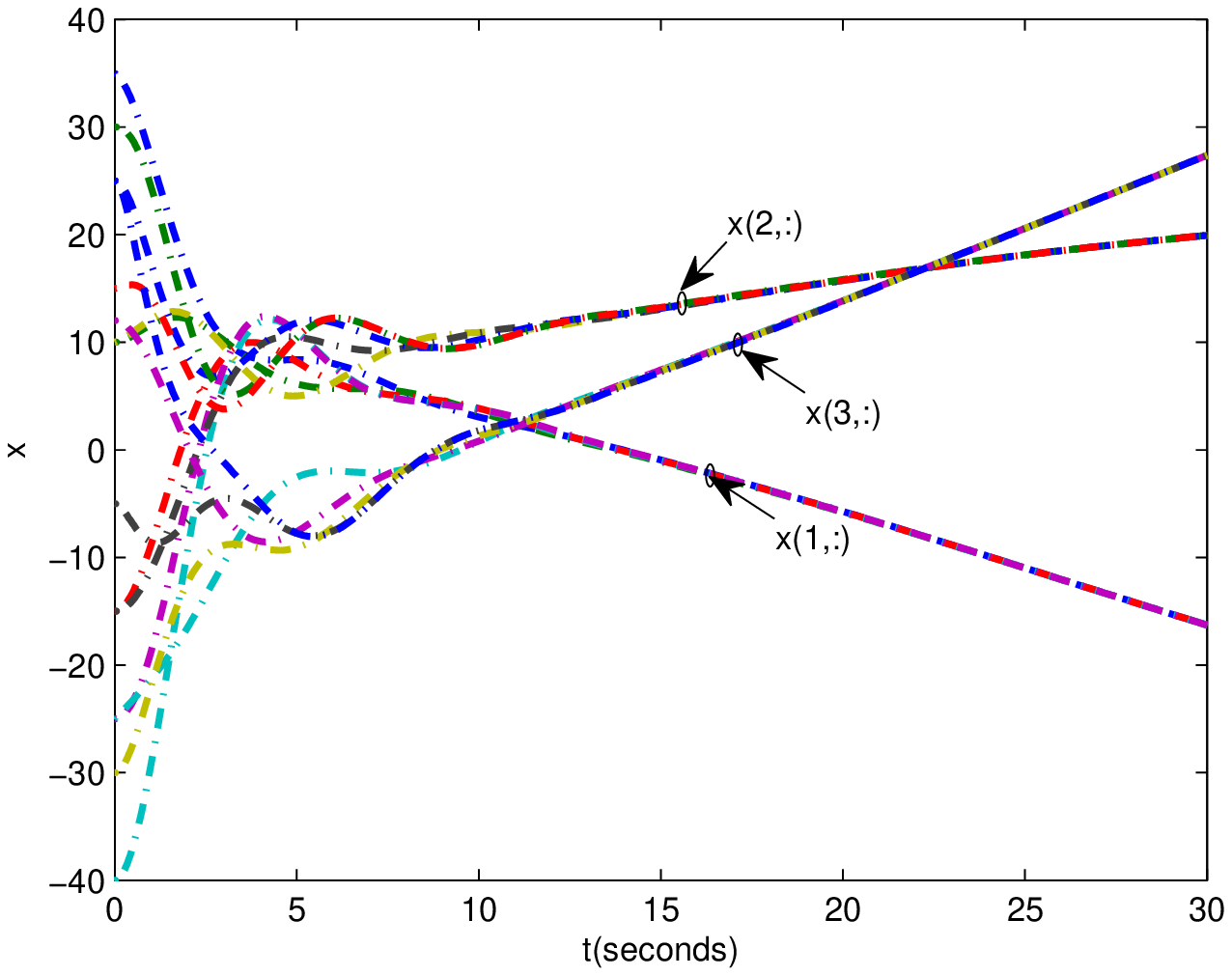}\label{fig:qux}}
{\includegraphics[width=65mm]{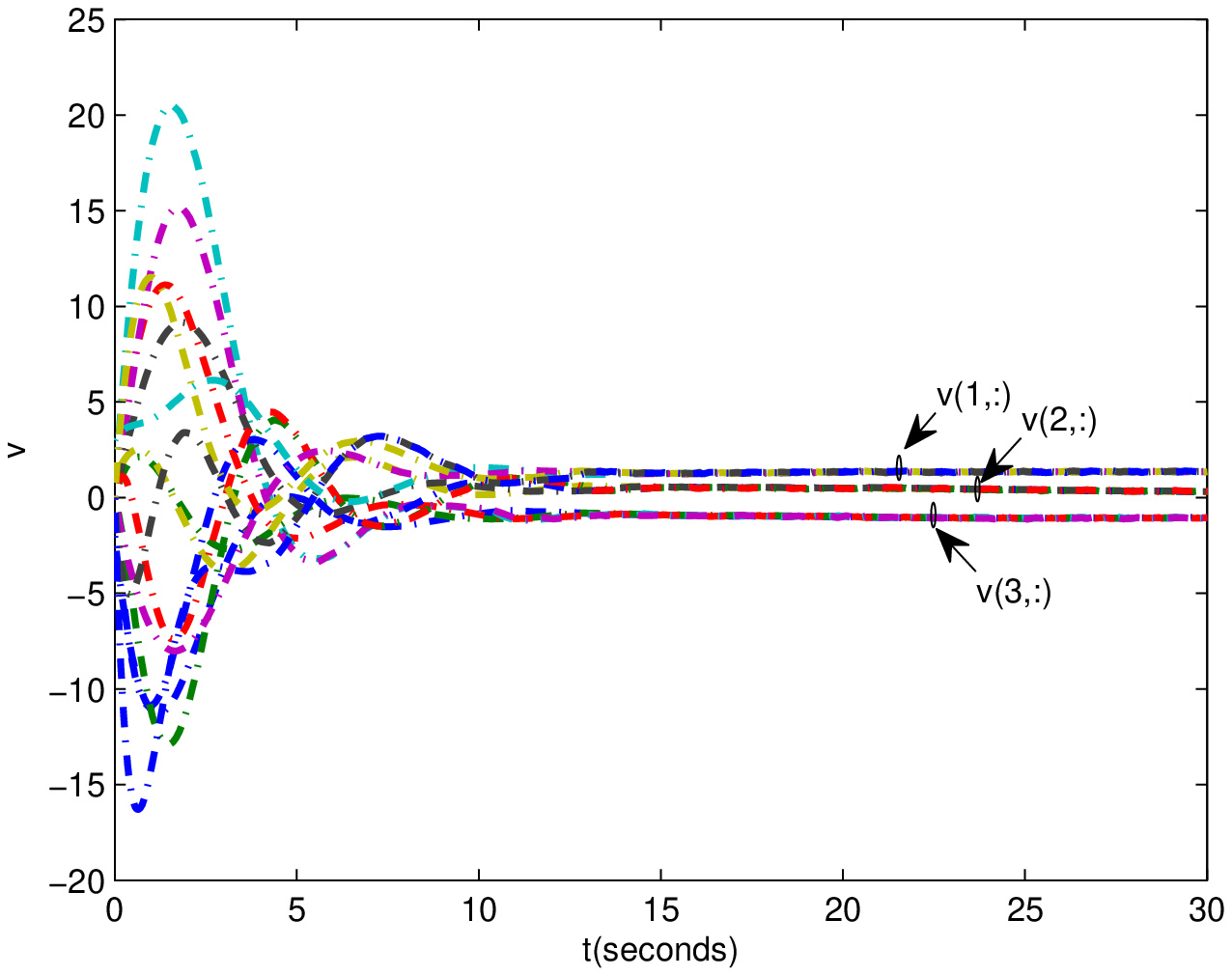}\label{fig:quv}}
}
\mbox{
{\includegraphics[width=65mm]{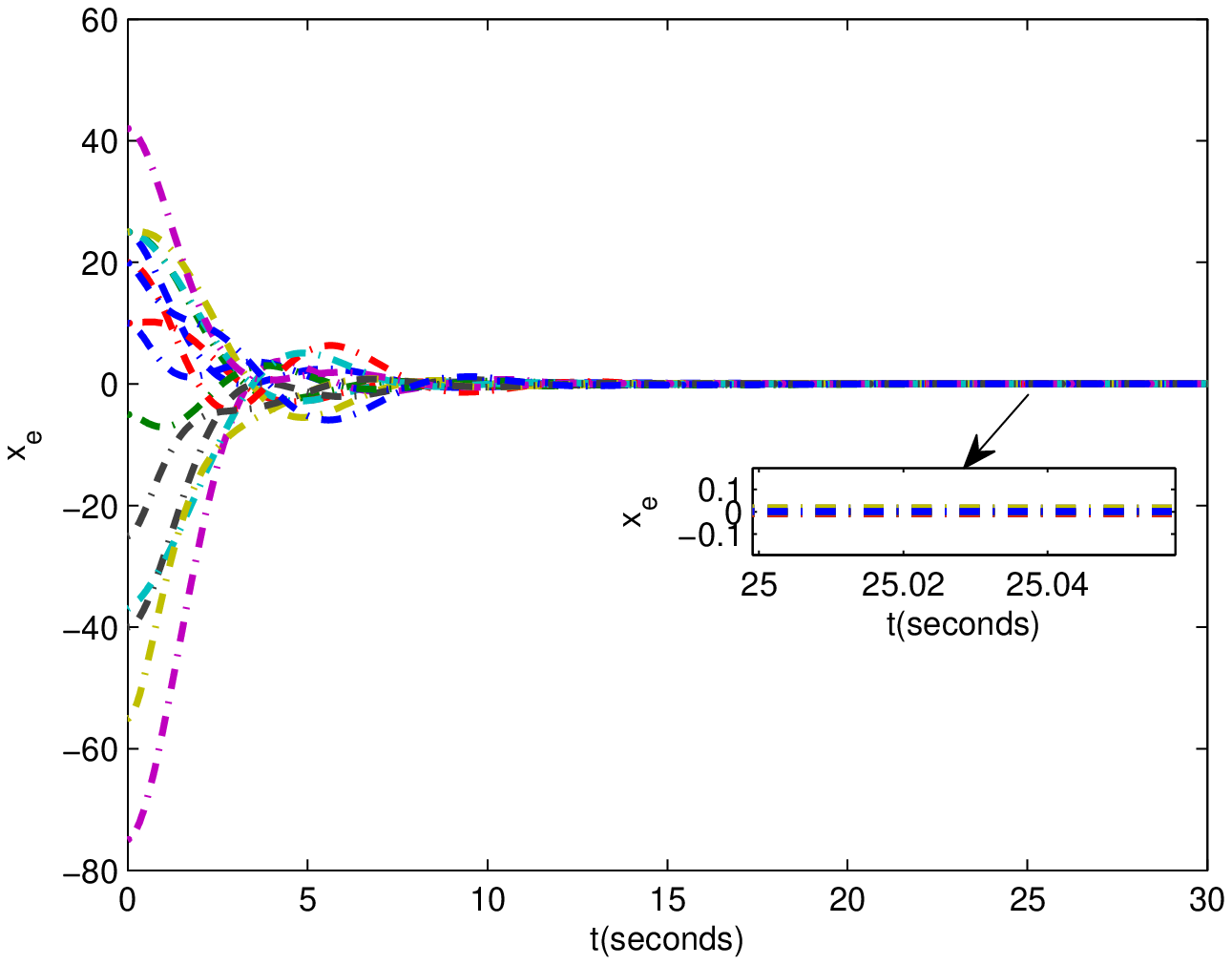}\label{fig:quxe}}
{\includegraphics[width=65mm]{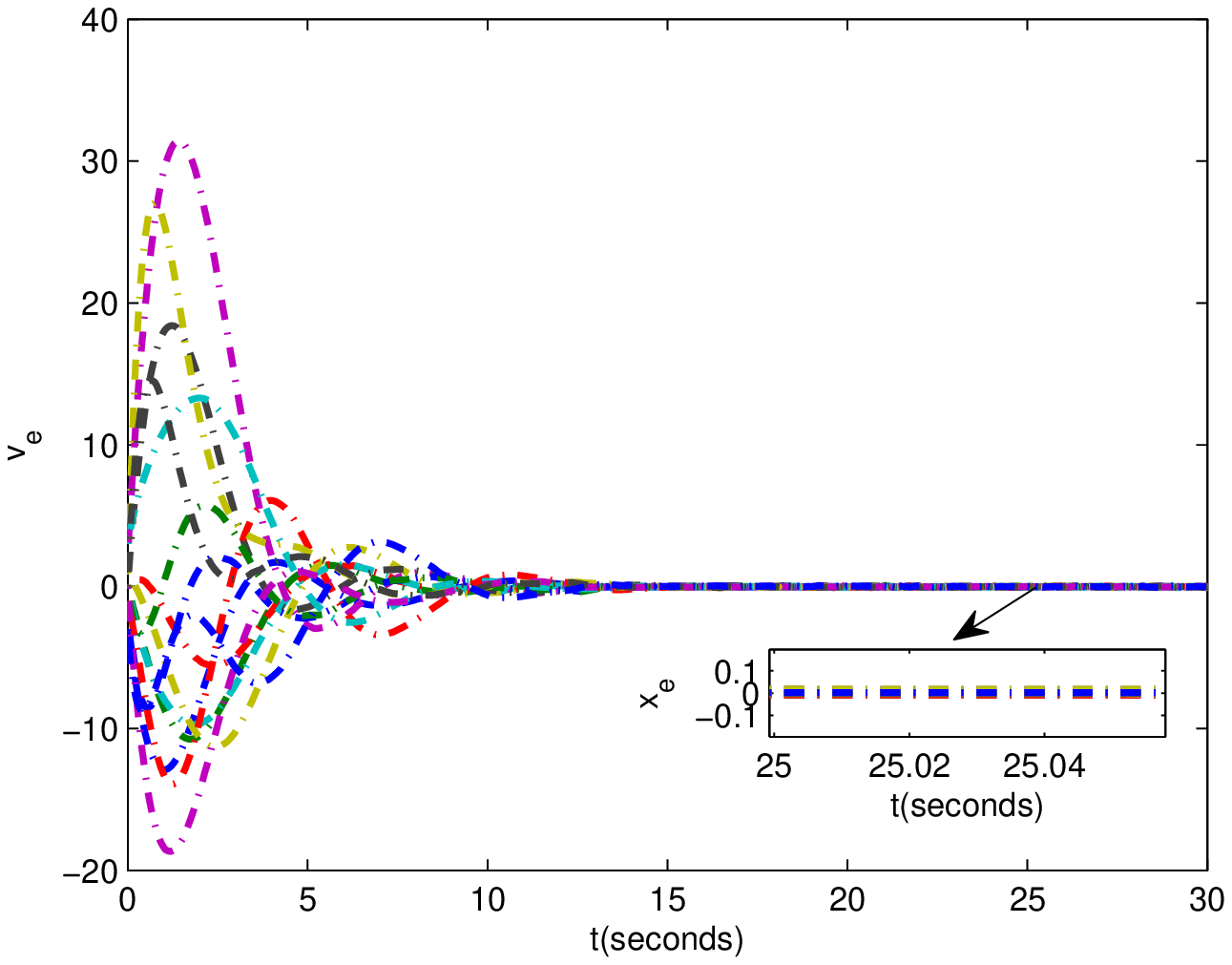}\label{fig:quve}}
}
\caption{Edge agreement under dynamic uniform quantizer.}
\label{fig:qu}
\end{center}
\end{figure*}
The dynamics of the $i$-th agent is described as \eqref{dynam}, in which $x_i\left(t\right), v_i\left(t\right), u_i\left(t\right) \in \rtn^3$. Through a simple calculation, we can obtain
$$
T = {\begin{pmatrix}\begin{smallmatrix}
    -1.00 \cr
    -1.00  \cr
    0.00   \cr
    -1.00   \cr
 \end{smallmatrix}\end{pmatrix} },~
R = {\begin{pmatrix}\begin{smallmatrix}
   1.00   &      0.00     &    0.00      &   0.00  & -1.00 \cr
   0.00   &     1.00     &    0.00      &   0.00 &  -1.00  \cr
    0.00  &       0.00  &  1.00    &     0.00    &     0.00   \cr
    0.00  &       0.00  &       0.00 &   1.00  &  -1.00   \cr
 \end{smallmatrix}\end{pmatrix} }.
$$
Suppose that the weighted diagonal matrix is defined as $\mathcal{W}=diag\{0.12,0.24,0.44,0.43,0.09\}$. By choosing $\sigma = 1.64$, we have
$${\hat L}_e = {\begin{pmatrix}\begin{smallmatrix}
    0.21  &  0.09   &      0.00  &  0.09   \cr
   -0.12  &  0.24   &      0.00  &  0.00  \cr
    0.00  & -0.24  &  0.44   &      0.00   \cr
    0.00  & -0.24  &       0.00   & 0.43   \cr
 \end{smallmatrix}\end{pmatrix} }$$
$$ {\hat L}_{_{\mathcal O}}=  {\begin{pmatrix}\begin{smallmatrix}
    0.12 &  0.00  &  0.00 &  0.00 &  -0.09 \cr
     -0.12 &   0.24 &  0.00  &  0.00  &  0.00 \cr
    0.00 &  -0.24  &  0.44 &  0.00  & -0.00 \cr
     0.00  & -0.24 &  0.00 &   0.43  &  0.00 \cr
 \end{smallmatrix}\end{pmatrix} } .
$$
Solving the Lyapunov equation \eqref{edgelyap} leads to
$$
 {H}=  {\begin{pmatrix}\begin{smallmatrix}
    2.47  &  0.16 &   0.07 &  -0.26 \cr
     0.16 &   2.86 &   0.39 &   0.45 \cr
    0.07  &  0.39  &  1.14 &  -0.01 \cr
     -0.26&    0.45 &  -0.01  &  1.22 \cr
 \end{smallmatrix}\end{pmatrix} }.
$$
Directed calculation yields $\lambda_{max}(\mathcal{P})=8.098$, $\lambda_{min}(\mathcal{P})=0.6157$, $\left\| {{\mathcal P} } \right\| =8.098$ and $\left\| {{\mathcal P}{{\mathcal L}_{_{{\mathcal T}1}}} } \right\|=6.7121$.
\begin{figure}[hbtp]
\begin{center}
{\includegraphics[width=60mm]{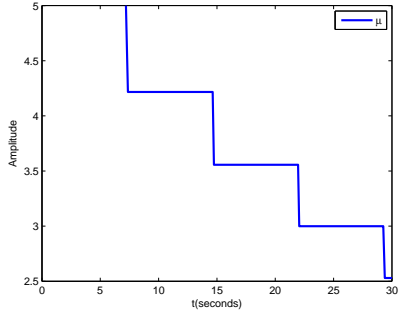}}
\caption{The amplitude of the zooming variable $\mu$.}
\label{fig:qu}
\end{center}
\end{figure}

Consider the quantized protocol \eqref{quantizedpro} with the dynamic uniform quantizer \eqref{dyq}. By choosing $\mathcal{M}= 63$ (i.e., only 7 bits information is required) with $\varepsilon = 0.75$, the condition \eqref{cond1} is satisfied. To ensure the initial condition lies in the ellipsoid $R_1$, $\mu_0$ can be chosen to be 10. The resulted zooming interval of the scheme proposed in this paper is $T = 6.2597s$. The simulation results with $\Delta = 0.1$ are shown in Fig. \ref{fig:qu}, from which we can see that $x_e(t)$ and $v_e(t)$ indeed converge to the equilibrium points asymptotically.


\section{Conclusion}
In this paper, we explored the edge agreement problem of second-order multi-agent system under quantized communication. Based upon the essential edge Laplacian, we derived a model reduction representation of the closed-loop multi-agent system for directed graph. Then, the dynamic quantized communication strategy based on the rooming in-rooming out scheme with finite quantization level was proposed. Through certainty equivalent quantized feedback controller and state transformation, the asymptotic stability of second-order multi-agent system under dynamic quantized effects can be guaranteed.

\section*{Acknowledgements}
This work was supported by the National Natural Science Foundation of China under grant 61403406.

\end{spacing}
\bibliographystyle{elsarticle-num}        
\bibliography{autosam}           
\end{document}